\documentclass[a4paper, 11pt]{article}

\usepackage{amsmath, amssymb, amsthm, multirow, url, fullpage}

\theoremstyle{plain}

\newtheorem{proposition}{Proposition}
\newtheorem{theorem}{Theorem}

\newtheorem*{claim*}{Claim}

\theoremstyle{definition}

\newcommand{\GF}{\mathrm{GF}}

\newcommand{\Sym}{\mathrm{Sym}}
\newcommand{\Alt}{\mathrm{Alt}}
\newcommand{\GL}{\mathrm{GL}}
\newcommand{\SL}{\mathrm{SL}}
\newcommand{\PSL}{\mathrm{PSL}}

\begin{document}

\title{Computing in matrix groups without memory}
\author{Peter J. Cameron\footnote{School of Mathematical Sciences, Queen Mary, University of London,  Mile End Road, London E1 4NS, UK and
	School of Mathematics and Statistics, University of St Andrews,	Mathematical Institute,	North Haugh, St Andrews, Fife KY16 9SS, UK.
   email: p.j.cameron@qmul.ac.uk, pjc@mcs.st-andrews.ac.uk},~Ben Fairbairn\footnote{Department of Economics, Mathematics and Statistics, Birkbeck, University of London, Malet Street, London WC1E 7HX, UK.	email: bfairbairn@ems.bbk.ac.uk}~and Maximilien Gadouleau\footnote{School of Engineering and Computing Sciences, Durham University, South Road, Durham, DH1 3LE, UK. email:m.r.gadouleau@durham.ac.uk}}

\maketitle

\begin{abstract}
Memoryless computation is a novel means of computing any function of a set of registers by updating one register at a time while using no memory. We aim to emulate how computations are performed on modern cores, since they typically involve updates of single registers. The computation model of memoryless computation can be fully expressed in terms of transformation semigroups, or in the case of bijective functions, permutation groups. In this paper, we view registers as elements of a finite field and we compute linear permutations without memory. We first determine the maximum complexity of a linear function when only linear instructions are allowed. We also determine which linear functions are hardest to compute when the field in question is the binary field and the number of registers is even. Secondly, we investigate some matrix groups, thus showing that the special linear group is internally computable but not fast. Thirdly, we determine the smallest set of instructions required to generate the special and general linear groups. These results are important for memoryless computation, for they show that linear functions can be computed very fast or that very few instructions are needed to compute any linear function. They thus indicate new advantages of using memoryless computation.
\end{abstract}


\section{Introduction}

\subsection{Memoryless computation}

Typically, swapping the contents of two variables $x$ and $y$ requires a buffer $t$, and proceeds as follows (using pseudo-code):
\begin{eqnarray*}
    t &\gets& x\\
    x &\gets& y\\
    y &\gets& t.
\end{eqnarray*}
However, the famous XOR swap (when $x$ and $y$ are sequences of bits), which we view in general as addition over a vector space:
\begin{eqnarray*}
    x &\gets& x+y\\
    y &\gets& x-y\\
    x &\gets& x-y,
\end{eqnarray*}
performs the swap without any use of memory. 

While the example described above is folklore in Computer Science, the idea to compute functions without memory was developed in \cite{Bur96,Bur04,BGT09,BM00,BM04a,BM04} and then independently rediscovered and developed in \cite{GR11a}. Amongst the results derived in the literature is the non-trivial fact that any function can be computed using memoryless computation. Moreover, only a number of updates linear in the number of registers is needed: any function of $n$ variables can be computed in at most $4n-3$ updates (a result proved for the Boolean alphabet in \cite{BGT09}, then independently extended to any alphabet in \cite{GR11a} and \cite{BGT13}), which reduces to $2n-1$ if the function is bijective. Memoryless computation has the potential to speed up computations not only by avoiding time-consuming communication with the memory but also by effectively combining the values contained in  registers. This indicates that memoryless computation can be viewed as an analogue in computing to network coding \cite{ACLY00,YLCZ06}, an alternative to routing on networks. It is then shown in \cite{GR11a} that for certain manipulations of registers, memoryless computation uses arbitrarily fewer updates than traditional, ``black-box'' computing. 

\subsection{Model for computing in matrix groups without memory}

In this paper, we are interested in computing linear bijective functions without memory. Some results already appear in \cite{GR11a} about these functions. For instance, any linear function can be computed in at most $2n-1$ updates; in this paper, we lower that upper bound to $\lfloor 3n/2\rfloor$, which is tight. The number of updates required to compute any manipulation of variables is also determined in \cite[Theorem 4.7]{GR11a}. 

Foremost, let us recall some notations and results from \cite{GR11a}. Let $A:= \GF(q)$ be a finite field (the {\em alphabet}) and $n \ge 2$ be an integer representing the number of registers (also called variables) $x_1,\ldots,x_n$. We denote $[n] = \{1,2,\ldots,n\}$. The elements of $A^n$ are referred to as {\em states}, and any state $a \in A^n$ is expressed as $a= (a_1,\ldots,a_n)$. For any $1 \le k \le n$, the $k$-th unit state is given by $e^k = (0,\ldots,0,1,0,\ldots,0)$ where the $1$ appears in coordinate $k$. We also denote the all-zero state as $e^0$.

For any $f \in \Sym(A^n)$, we denote its $n$ coordinate functions as $f_1,\ldots,f_n : A^n \to A$, i.e.
$f(x) = (f_1(x), \ldots,f_n(x))$
for all $x = (x_1,\ldots,x_n) \in A^n$. We say that the $i$-th coordinate function is {\em trivial} if it coincides with that of the identity: $f_i(x) = x_i$; it is nontrivial otherwise.

A bijective {\em instruction} is a permutation $g$ of $A^n$ with one nontrivial coordinate function:
$$
    g(x) = (x_1,\ldots,x_{j-1},g_j(x),x_{j+1},\ldots,x_n)
$$
for some $1 \le j \le n$. We say the instruction $g$ {\em updates} the $j$-th coordinate. We can represent this instruction as
$$
    y_j \gets g_j(y)
$$
where $y = (y_1,\ldots,y_n) \in A^n$ represents the contents of the registers. A {\em program} computing $f$ is simply a sequence of instructions whose combination is $f$; the instructions are typically denoted one after the other.

With this notation, the swap of two variables can be viewed as computing the permutation $f$ of $A^2$ defined as $f(x_1,x_2) = (x_2,x_1)$, and the program is given by
\begin{eqnarray*}
    y_1 &\gets& y_1 + y_2 \qquad (= x_1 + x_2)\\
    y_2 &\gets& y_1 - y_2 \qquad (= x_1)\\
    y_1 &\gets& y_1 - y_2 \qquad (= x_2).
\end{eqnarray*}

In this paper, we want to compute a linear transformation $f: A^n \to A^n$, i.e.
$$
    f(x) = xM^\top
$$
for some matrix $M \in A^{n \times n}$. We denote the rows of $M$ as $f_i$. We restrict ourselves to linear instructions only, i.e. instructions of the form
$$
    y_i \gets v \cdot y = \sum_{j=1}^n v_j y_j,
$$
for some $v = (v_1,\ldots, v_n) \in A^n$. In particular, the instruction above is a permutation if and only if $v_i \ne 0$. Note that computing $f$ without memory is then equivalent to computing $M$ by starting from the identity matrix and updating one row at a time.

The set $\mathcal{M}(\GF(q)^n)$ of bijective linear instructions then corresponds to the set of nonsingular matrices with at most one nontrivial row:
$\mathcal{M} = \{S(i,v) : 1 \le i \le n, v \in A^n(i)\},$
where
\begin{eqnarray*}
    A^n(i) &=& \{v \in A^n, v_i \neq 0\} \,\mbox{for all}\, 1 \le i \le n,\\
    S(i,v) &=& \left(\begin{array}{c|c|c}
        I_{i-1} & \multicolumn{2}{c}{0}\\
        \hline
        \multicolumn{3}{c}{v}\\
        \hline
        \multicolumn{2}{c|}{0} & I_{n-i}
    \end{array}\right) \in A^{n \times n}.
\end{eqnarray*}
We remark that $S(i,v)^{-1} = S(i,-v_i^{-1}v)$ for all $i,v$.

Following \cite{CFG12}, we say a group is {\em internally computable} if it can be generated by its instructions, i.e. if any element of the group can be computed by a program using instructions from $G$. For instance, Gaussian elimination proves that $\GL(n,q)$ is internally computable. We prove in Proposition \ref{prop:SL} that $\SL(n,q)$ is also internally computable. For any internally computable group $G$, two main problems arise. First, we want to know how fast we can compute any element of $G$: we will prove that the maximum complexity in the general linear group is $\lfloor 3n/2 \rfloor$ instructions in Theorem \ref{th:diameter_GL}. More surprisingly, if $q=2$ and $n$ is even, then the matrices requiring $3n/2$ instructions are fully characterised in Proposition \ref{prop:GL(2m,2)}. Note that the average complexity over all elements of a group is also interesting; for $\GL(n,q)$, this quantity tends to $n$ instructions when $q$ is large \cite{GR11a}. 

Secondly, due to the large number of possible instructions, it seems preferable to work with restricted sets of instructions which could be efficiently used by a processor. Therefore, we also want to know the minimum number of instructions required to generate the whole group. We shall determine this for the special and general linear groups in Theorems \ref{th:SL} and \ref{th:GL}, respectively. The fact that it is equal to $n$ in most cases--and $n+1$ otherwise--shows how easy it is to compute linear functions without memory and how little space would be required to store those minimal sets of instructions. 

For any internally computable group $G$ and any $g \in G$, we denote the shortest length of a program computing $g$ using only instructions from $G$ as $\mathcal{L}(g,G)$; we refer to this quantity as the {\em complexity} of $g$ in $G$. If $H \le G$ and $\mathcal{L}(h,H) = \mathcal{L}(h,G)$ for all $h \in H$, we say that $H$ is {\em fast} in $G$. It is still unknown whether $\GL(n,q)$ is fast in $\Sym(\GF(q)^n)$, i.e. if we cannot compute linear functions any faster by allowing non-linear instructions. However, we will prove in Proposition \ref{prop:SL} that the special linear group is not fast in the general linear group (unless $q=2$).

We would like to emphasize that we only consider bijective linear functions, i.e. computing in matrix groups. The case of any bijective function is studied in \cite{CFG12}, where analogue results are derived for the symmetric and alternating groups of $A^n$ ($A$ being any finite set of cardinality at least $2$).

The rest of the paper is organised as follows. In Section \ref{sec:complexity}, we determine the maximum complexity of any matrix in $\GL(n,q)$ and investigate which matrices have highest complexity. Then, in Section \ref{sec:matrix_groups}, we determine whether some matrix groups are internally computable, and we show that $\SL(n,q)$ is internally computable but not fast in $\GL(n,q)$. Finally, in Section \ref{sec:generating}, we determine the minimum size of a generating set of instructions for both the special and general linear groups.

\section{Maximum complexity in the general linear group} \label{sec:complexity}

\begin{theorem} \label{th:diameter_GL}
Any matrix in $\GF(q)^{n \times n}$ can be computed in at most $\lfloor 3n/2 \rfloor$ linear instructions. This bound is tight and reached for some matrices in $\GL(n,q)$.
\end{theorem}

\begin{proof}
We consider the general case where the matrix $M$ we want to compute is not necessarily invertible. We prove the statement by strong induction on $n \ge 1$; it is clear for $n=1$. Suppose it holds for up to $n-1$.

For any $S \subset [n]$, we refer to the matrix $M_S \in \GF(q)^{|S| \times |S|}$ with entries $M(i,j)$ for all $i,j \in S$ as the $S$-principal of $M$. Suppose that $M$ has a nonsingular $S$-principal $M_S$, say $S = \{1,\ldots,k\}$ and express $M$ as
$M = \left(\begin{array}{c|c}
    M_S & N\\
    \hline
    P & Q
    \end{array}\right),$
where $N \in \GF(q)^{k \times n-k}$, $P \in \GF(q)^{n-k \times k}$, $Q \in \GF(q)^{n-k \times n-k}$. We give a program for $M$ in two main steps and no more than $\left\lfloor 3n/2 \right\rfloor$ instructions.

The first step computes $(M_S|N)$. By hypothesis, $M_S$ can be computed in $\left\lfloor 3k/2 \right\rfloor$ instructions. We can easily convert that program in order to compute the matrix $(M_S | N)$ as follows. Consider the final update of row $j$: $y_j \gets f_j$ (i.e., the $j$-th row must be equal to that of $M$ after its last update). The $j$-th row of $N$, say $n_j$ is a linear combination of the rows of $(0|I_{n-k})$, hence simply replace $y_j \gets f_j$ by $y_j \gets f_j + n_j$ and in any subsequent instruction, replace every occurrence of $y_j$ by $y_j - n_j$.

The second step computes $(P|Q)$. Note that the rows $p_1,\ldots,p_{n-k}$ of $P$ can be expressed as linear combinations of those of $M_S$: $P = RM_S$ where the rows of $R = PM_S^{-1} \in \GF(q)^{n-k \times k}$ are denoted $r_1,\ldots,r_{n-k}$. By hypothesis, the matrix $X := Q - RN$ (with rows $x_1,\ldots,x_{n-k}$) can be computed in $\left\lfloor 3(n-k)/2 \right\rfloor$ instructions. Again this can be converted to compute $(P|Q)$ as follows. Suppose $i$ is the first row to have its last update in a program computing $X$, say it is $y_i \gets \sum_{l=1}^{n-k} a_{i,l} y_l$. Then the new program for $(P|Q)$ is 
$$
	y_{k+i} \gets \sum_{l=k+1}^n a_{i,l} y_l + \sum_{l=1}^k r_{i,l} y_l = (r_iM_S | x_i + r_i N) = (p_i | q_i).
$$
Then replace every future occurrence of $y_i$ with $y_{k+i} - \sum_{l=1}^k r_{i,l} y_l$. Suppose that $i'$ is the next row to have its last update $y_{i'} \gets \sum_{l=1}^{n-k} a_{i',l} y_l$; this is converted to
$$
	y_{k+i'} \gets \sum_{l=k+1}^n a_{i',l} y_l - a_{i',i} \sum_{l=1}^k r_{i,l} y_l + \sum_{l=1}^k r_{i',l} y_l = (r_{i'} M_S | x_{i'} + r_{i'} N) = (p_{i'} | q_{i'}).
$$
Again, every future occurrence of $i'$ will be replaced with $y_{k+i'} - \sum_{l=1}^k r_{i',l} y_l$, and so on. By induction, we can then easily prove that this converted program computes $(P | Q)$.

Now suppose $M$ does not have any invertible principal. Let $D$ be the directed graph whose adjacency matrix $A_D$ satisfies $A_D(i,j) = 1$ if $M(i,j) \ne 0$ and $A_D(i,j) = 0$ if $M(i,j) = 0$. If $D$ is acyclic, then $M$ can be computed in $n$ instructions, for it is (up to renaming the vertices in topological order) an upper triangular matrix with zeros on the diagonal. Otherwise, $D$ has girth $n$, for otherwise the adjacency matrix of the subgraph induced by a shortest cycle forms a nonsingular principal. Therefore $D$ is a cycle, and $M$ can be computed in $n+1$ instructions by \cite[Proposition 4.6]{GR11a}.

The tightness of the bound follows from \cite[Corollary 2]{GR11a}.
\end{proof}

By the proof of Theorem \ref{th:diameter_GL}, we see that the only matrices in $\GL(2,q)$ which are a product of three instructions are exactly those whose support is the permutation matrix of a transposition. Proposition \ref{prop:GL(2m,2)} below extends this result to any even order when the matrices are over $\GF(2)$.

\begin{proposition} \label{prop:GL(2m,2)}
In $\GL(2m,2)$, the only matrices which are the product of no fewer than $3m$ instructions are the permutation matrices of fix-point free involutions.
\end{proposition}

\begin{proof}
We prove it by strong induction on $m$; it is clear for $m=1$ and checked by computer for $m=2$, therefore we assume $m \ge 3$ and that it holds for up to $m-1$. For any $k \ge 1$, we denote the permutation matrix of $(1,2)\cdots(2k-1,2k)$ as $J_k$. We say that two matrices $M$ and $N$ are equivalent if $M = \Pi N \Pi^{-1}$ for some permutation matrix $\Pi$.

Let $M \in \GL(2m,2)$ be a matrix at distance $3m$ from the identity which is not equivalent to $J_m$. According to the proof of Theorem \ref{th:diameter_GL}, the graph $D$ with adjacency matrix $M$ must contain a directed cycle of length $< 2m$. The graph $D$ has girth $2$, for otherwise there is a principal of size other than $2$ and hence $M$ can be computed in fewer than $3m$ instructions by using the two-step algorithm in the proof of Theorem  \ref{th:diameter_GL}. More generally, any invertible principal of $M$ must have even size and be a conjugate of $J_k$ for some $k$.

Hence we can express $M$ (up to equivalence) as
$M = \left(\begin{array}{c|c}
    J_1 & N\\
    \hline
    P & Q
    \end{array}\right),$
where $N \in \GF(2)^{2 \times 2(m-1)}$, $P \in \GF(2)^{2(m-1) \times 2}$, $Q \in \GF(2)^{2(m-1) \times 2(m-1)}$. By the same argument, we can first compute $J_1$ and then the matrix $Q + PJ_1N$, hence these matrices must satisfy (up to equivalence) $PJ_1N + Q = J_{m-1}$.

Since $M \ne J_m$, there exists $2 \le k \le m$ such that the $\{1,2,2k-1,2k\}$-principal of $M$ is not equal to $J_2$. The conditions above mean that this principal is not invertible, neither is any of its $T$-principals for $|T| = 3$, and it can be expressed as
$$
    \begin{pmatrix}
    0 & 1 & a & b\\
    1 & 0 & c & d\\
    e & f & 0 & \alpha\\
    g & h & \beta & 0
    \end{pmatrix},
$$
where $\alpha = bf + de + 1$ and $\beta = ah + cg + 1$. However, it can be verified that no such matrix exists.
\end{proof}

We remark that the situation for $\GL(2m+1,2)$ is much more
complicated. Indeed, the permutation matrices of $(1,2)(3,4)\cdots(2m-1,2m,2m+1)$ and its conjugates are still extremal, but many other matrices are
also extremal. For example by Theorem 2 we know that the diameter of
the Cayley graph for $\GL(3,2)$ is 4 and clearly there are only two
extremal permutation matrices in $\GL(3,2)$ however there are 35 matrices equal to the
product of 4 and no fewer linear instructions in this group - see Table
\ref{GL3}.

\begin{table}
\begin{center}
\begin{tabular}{cccccc}
$\begin{pmatrix} 0&1&0\\ 0&0&1\\ 1&0&0 \end{pmatrix}$
&
$ \begin{pmatrix} 0&0&1\\ 1&0&0\\ 0&1&0 \end{pmatrix}$
& & & & \\
\\
$\begin{pmatrix} 0&1&0\\ 1&0&0\\ 1&0&1\end{pmatrix}$
&
$\begin{pmatrix} 0&1&0\\ 1&0&0\\ 0&1&1\end{pmatrix}$
&
$ \begin{pmatrix} 1&0&1\\ 0&0&1\\ 0&1&0 \end{pmatrix}$
&
$\begin{pmatrix} 1&1&0\\ 0&0&1\\ 0&1&0 \end{pmatrix}$
&
$\begin{pmatrix} 0&0&1\\ 0&1&1\\ 1&0&0 \end{pmatrix}$
&
$\begin{pmatrix} 0&0&1\\ 1&1&0\\ 1&0&0\end{pmatrix}$
\\
\\
$\begin{pmatrix} 0&1&1\\ 1&0&0\\ 0&1&0\end{pmatrix}$
&
$\begin{pmatrix} 0&1&1\\ 0&0&1\\ 1&0&0\end{pmatrix}$
&
$\begin{pmatrix} 0&0&1\\ 1&0&0\\ 1&1&0\end{pmatrix}$
&
$\begin{pmatrix} 0&0&1\\ 1&0&1\\ 0&1&0\end{pmatrix}$
&
$\begin{pmatrix} 0&0&1\\ 0&1&0\\ 1&1&0 \end{pmatrix}$
&
$\begin{pmatrix} 0&1&0\\ 1&0&1\\ 1&0&0\end{pmatrix}$
\\
\\
$\begin{pmatrix} 0&1&0\\ 0&1&1\\ 1&0&0\end{pmatrix}$
&
$\begin{pmatrix} 0&1&0\\ 0&0&1\\ 1&0&1\end{pmatrix}$
&
$\begin{pmatrix} 1&0&1\\ 1&0&0\\ 0&1&0\end{pmatrix}$
&
$\begin{pmatrix} 0&0&1\\ 1&1&0\\ 0&1&0\end{pmatrix}$
&
$\begin{pmatrix} 1&1&0\\ 0&0&1\\ 1&0&0\end{pmatrix}$
&
$\begin{pmatrix} 0&0&1\\ 1&0&0\\ 0&1&1\end{pmatrix}$
\\
\\
$\begin{pmatrix} 0&0&1\\ 1&0&1\\ 1&1&0\end{pmatrix}$
&
$\begin{pmatrix} 0&1&0\\ 1&0&1\\ 1&1&0 \end{pmatrix}$
&
$\begin{pmatrix} 0&1&1\\ 0&0&1\\ 1&1&0\end{pmatrix}$
&
$\begin{pmatrix} 0&1&1\\ 1&0&0\\ 1&1&0\end{pmatrix}$
&
$\begin{pmatrix} 0&1&1\\ 1&0&1\\ 1&0&0\end{pmatrix}$
&
$\begin{pmatrix} 0&1&1\\ 1&0&1\\ 0&1&0\end{pmatrix}$
\\
\\
$\begin{pmatrix} 1&1&1\\ 0&0&1\\ 0&1&0\end{pmatrix}$
&
$\begin{pmatrix} 0&1&0\\ 1&0&0\\ 1&1&1\end{pmatrix}$
&
$\begin{pmatrix} 0&0&1\\ 1&1&1\\ 1&0&0\end{pmatrix}$
\\
\\
$\begin{pmatrix} 1&1&0\\ 1&1&1\\ 1&0&0\end{pmatrix}$
&
$\begin{pmatrix} 0&0&1\\ 1&1&1\\ 0&1&1\end{pmatrix}$
&
$\begin{pmatrix} 1&1&1\\ 1&1&0\\ 0&1&0\end{pmatrix}$
&
$\begin{pmatrix} 1&1&1\\ 0&0&1\\ 1&0&1\end{pmatrix}$
&
$\begin{pmatrix} 1&0&1\\ 1&0&0\\ 1&1&1\end{pmatrix}$
&
$\begin{pmatrix} 0&1&0\\ 0&1&1\\ 1&1&1\end{pmatrix}$
\end{tabular}
\end{center}
\caption{The 35 matrices in GL(3,2) that are a product of four
linear instructions and their orbits under the action of $\Sym(3)$.}\label{GL3}
\end{table}

\section{Some matrix groups} \label{sec:matrix_groups}

We first discuss the special linear groups. Recall that a {\em transvection} is any permutation $t_{\phi,v}$ of $A^n$, such that
$$
    t_{\phi,v}(x) = x + \phi^\top v x^\top
$$
for all $x \in \GF(q)^n$, where $v, \phi \in \GF(q)^n$, \cite{CH91}. Then $t_{\phi,v}$ is an instruction if and only if $\phi$ (viewed as a
column vector) has only one nonzero coordinate. In other words, any transvection which is an instruction is represented by a shear matrix $S(i,e^i + ae^j)$ for some $i,j$ and $a \in \GF(q)$.

\begin{proposition} \label{prop:SL} ($i$) The group $\SL(n,q)$ is internally computable
for any $n$ and prime power $q$.

($ii$) If $q \ne 2$ then $\SL(n,q)$ is not fast in $\GL(n,q)$.
\end{proposition}

\begin{proof}
($i$) This is simply the observation that any transvection is a product of instructions and the transvections are well known to generate the special
linear group - see for instance \cite[p.45]{Wil09}.

($ii$) We prove this in the case $n=2$, the extension to the general
case being clear. If $q\not=2$ then there exists an element
$\alpha \in \GF(q)$ such that $\alpha\not=0,1$. Inside
$\GL(2,q)$ we thus have
$$\begin{pmatrix}
\alpha&0\\
0&\alpha^{-1}
\end{pmatrix}=
\begin{pmatrix}
\alpha&0\\
0&1
\end{pmatrix}
\begin{pmatrix}
1&0\\
0&\alpha^{-1}
\end{pmatrix}
$$
which expresses the above element of $\SL(2,q)$ as a
product of two instructions. Inside $\SL(2,q)$ however we
have that
\begin{align*}
	\begin{pmatrix}
	1+xy&x\\
	y&1
	\end{pmatrix} &=
	\begin{pmatrix}
	1&x\\
	0&1
	\end{pmatrix}
	\begin{pmatrix}
	1&0\\
	y&1
	\end{pmatrix},\\
	\begin{pmatrix}
	1&x\\
	y&1+xy
	\end{pmatrix} &=
	\begin{pmatrix}
	1&0\\
	y&1
	\end{pmatrix}
	\begin{pmatrix}
	1&x\\
	0&1
	\end{pmatrix}
\end{align*}
for any $x,y\in\GF(q)$. Since $\alpha\not=1$ the original
matrix cannot be of this form and thus cannot be expressed as a
product of just two instructions inside $\SL(2,q)$.
\end{proof}

The argument in the proof of ($ii$) can be easily generalised to show that any subgroup of $\GL$ defined as the set of matrices with determinant in a proper subgroup of the multiplicative group of $\GF(q)$ is not fast.

We remark that if $q=2$ then $\SL(n,q)=\GL(n,q)$. Unfortunately most
other groups that are naturally matrix groups are not internally
computable in their natural $\GF(q)$ modules.

\begin{proposition}
Orthogonal groups of type +, unitary and symplectic groups are not internally computable.
\end{proposition}

\begin{proof}
In the orthogonal and unitary cases this is simply the observation
that a matrix $A$ is an element of these groups if it satisfies
$AA^\top = I$ or $A\bar{A}^\top = I$, respectively, where the bar indicates the automorphism of $\GF(q)$
of order $2$ when it exists \cite[p.66 \&
p.70]{Wil09}. Clearly no instruction satisfies either condition and
so these groups contain no instructions whatsoever.

Elements of the symplectic group Sp$(2n,q)$ are precisely the
invertible matrices of the form
$\left(\begin{array}{c|c}
A&B\\
\hline C&D
\end{array}\right)$
where $A$, $B$, $C$ and $D$ are $n\times n$ matrices such that
\begin{eqnarray*}
    AD^\top-BC^\top &=& I,\\
    AB^\top &=& A^\top B \quad \mbox{and}\\
    CD^\top &=& C^\top D.
\end{eqnarray*}
For an instruction to be of the above form one of $B$ or $C$ must be
the all zeros matrix and $A=D=I$. If $C=0$, we see that $B$ must be a matrix with only one nonzero entry, which lies on the diagonal; if $B=0$, we obtain its transpose. Therefore, the symplectic instructions generate a group of matrices where $A$,$B$, $C$ and $D$ are all diagonal; this is clearly a proper subgroup of Sp$(2n,q)$.
\end{proof}

\begin{proposition}
The groups $^2B_2(2^{2r+1})$, $^3D_4(q)$, $G_2(q)$,
$^2G_2(3^{2r+1})$ and $^2F_4(2^{2r+1})$ are not internally
computable.
\end{proposition}

\begin{proof}
We prove this in the case of $^2B_2(2^{2r+1})$ acting on its natural
4 dimensional $\GF(2^{2r+1})$ module the cases of $^2G_2(3^{2r+1})$
acting on its natural 7 dimensional $\GF(3^{2r+1})$ module and
$^2F_4(2^{2r+1})$ acting on its natural 26 dimensional
$\GF(2^{2r+1})$ module being entirely analogous. Furthermore
analogous arguments apply to $^3D_4(q)$ and $G_2(q)$ acting on their
natural 26 and 8 dimensional $\GF(q)$ modules respectively.

An instruction whose only non-zero off-diagonal entries are
contained entirely on the bottom row must be contained in the
subgroup of lower triangular matrices. The non-trivial elements of
this subgroup, however, are of the form
$$\left(\begin{array}{cccc}
1&0&0&0\\
\alpha\beta^{-1}&1&0&0\\
\alpha\beta&\beta^2&1&0\\
\alpha^2&0&\alpha\beta^{-1}&1
\end{array}\right)$$
where $\alpha\in\mbox{GF}(2^{2r+1})$ and $\beta=\alpha^{2^{r+1}-1}$
\cite[p.115]{Wil09}. Clearly this subgroup contains no instructions
and so the subgroup of $^2B_2(2^{2r+1})$ generated by any
instructions is a proper subgroup.
\end{proof}

\section{Generating linear groups} \label{sec:generating}

The purpose of this section is to determine the minimum number of instructions sufficient to generate some matrix groups. The reader is reminded of the elements $S(i,v)$ that we defined just before Theorem \ref{th:diameter_GL}. We also define the vectors $v^i \in \GF(q)^n$ such that $v^i = e^i + e^{i+1}$ for $i \le n-1$ and $v^n = e^1 + e^n$.


We first consider the special linear group.

\begin{theorem} \label{th:SL}
The group $\SL(n,q)$ is generated by $n$ instructions unless $n=2$, $q=2^m$ ($m \ge 2$), where it is generated by $3$ instructions.
\end{theorem}

\begin{proof}
The rest of the proof goes by induction on $n$, but we split the proof according to the parity of $q$. First, suppose $q$ is odd. An immediate consequence
of a classical Theorem incorrectly attributed to Dickson \cite{Dic01} (it was actually
proved by Wilman and Moore, see \cite[Corollary 2.2]{Kin06}) tells us that the maximal subgroups of
$\PSL(2,q)$, $q$ odd (these can easily be seen to ``lift'' to maximal subgroups of
$\SL(2,q)$) are all isomorphic to one of
\begin{itemize}
    \item[-] $\Alt(4)$, $\Sym(4)$ or $\Alt(5)$;
    \item[-] A dihedral group of order either $q+1$ or $q-1$;
    \item[-] A subfield subgroup;
    \item[-] A stabiliser of a one dimensional subspace in the action on the $q+1$ subspaces of $\GF(q)^2$ on which $\mathrm{(P)SL}(2,q)$ naturally acts.
\end{itemize}

Consider the matrices/instructions
$$
	S(1,(1,x)) = \begin{pmatrix} 1 & x\\ 0 & 1\end{pmatrix}, \quad S(2,(y,1)) = \begin{pmatrix} 1 & 0 \\  y & 1\end{pmatrix}.
$$
We prove that the group they generate does not belong to any of the maximal subgroups. First, the copies of Alt(4), Sym(4) and Alt(5).  In characteristic 3, the only way two elements of order 3 can be contained in a copy of Alt(4) or Sym(4) is if their product has order 1 or 3 (in which case they're contained in the same cyclic subgroup, which the above two matrices clearly are not) or 2 (and by direct calculation our two matrices do not have a product of order 2). Finally we can eliminate Alt(5) since this subgroup can only exist in characteristic 3 if $q=3$ or 9 which are easily eliminated by computer. In characteristic 5 there are no elements order 5 in Alt(4) and Sym(4) and for Alt(5) this maximal subgroup only exists when $q$ satisfies certain congruences that a power of 5 never satisfies. For characteristic  greater than $p>5$ there are clearly no elements of order $p$ in any of Alt(4), Sym(4) or Alt(5).

Since $p$ is coprime to
both $q+1$ and $q-1$ neither of these belong to a maximal dihedral subgroup.
The only one dimensional subspace fixed by the first matrix is spanned by
the (column) vector $(1,0)^\top$ whilst the second only fixes the subspace
spanned by the (column) vector $(0,1)^\top$, so no one dimensional subspace is fixed
by the subgroup these generate. Recall that the product of these two
matrices is the matrix
$\begin{pmatrix} 1+xy & x\\ y & 1 \end{pmatrix}$
which has trace $2+xy$. Choosing $x$ and $y$ so
that $2+xy$ is contained in no proper subfield of $\GF(q)$ now gives a pair of
elements that cannot generate a subfield subgroup. It follows that this
pair must generate the whole group.\\
\\
We now prove the inductive step. Let $x,y \in \GF(q)$ such that the instructions
$\begin{pmatrix} 1 & x\\ 0 & 1 \end{pmatrix},$ $\begin{pmatrix} 1 & 0\\ y & 1 \end{pmatrix}$
generate $\SL(2,q)$. Then we claim that the following set of $n$ instructions generates $\SL(n,q)$:
$$
    \{S(i,v^i) : 1 \le i \le n-2\} \cup \{S(n-1,e^{n-1} + xe^n), S(n, e^n + ye^1)\}.
$$
Let us remark that we can easily generate any instruction of the form $S(i,e^i + e^j)$ for $1 \le i < j \le n-1$ (and hence any of the form  $S(i,e^i - e^j)$ as well). We can then easily generate $S(i,e^i + xe^n)$ for any $1 \le i \le n-1$. We also generate any transvection of the form $S(n,e^n + y e^i)$ for any $1 \le i \le n-1$ as such: 
$$
    S(n,e^n + ye^i) = S(n,e^n-ye^1) S(1,e^1-e^i) S(n,e^n+ye^1) S(1,e^1+e^i).
$$
Displaying only the columns and rows indexed $1,i,n$, the equation above reads
$$
	\begin{pmatrix} 1 &0 &0\\ 0 &1 &0\\ 0 &y &1	\end{pmatrix} =
	\begin{pmatrix} 1 &0 &0\\ 0 &1 &0\\ -y &0 &1\end{pmatrix}
	\begin{pmatrix} 1 &-1 &0\\ 0 &1 &0\\ 0 &0 &1\end{pmatrix}
	\begin{pmatrix} 1 &0 &0\\ 0 &1 &0\\ y &0 &1\end{pmatrix}
	\begin{pmatrix} 1 &1 &0\\ 0 &1 &0\\ 0 &0 &1\end{pmatrix}.
$$
By combining the two types of transvections, we obtain all possible transvections of the type $S(i,e^i + ae^n)$ or $S(n,e^n + ae^i)$ for all $a \in \GF(q)$. We are done with the last coordinate, and we tackle the penultimate coordinate by considering
$$
    Q = \left(\begin{array}{c|cc}
    I_{n-2} & \multicolumn{2}{c}{0}\\
    \hline
    \multirow{2}{*}{0} & 0 & -1\\
     & 1 & 0
    \end{array}\right).
$$
Note that $Q$ is indeed generated by $S(n-1,e^{n-1} + xe^n)$ and $S(n,e^n + y e^{n-1})$. We then obtain the two required types of transvections:
\begin{align*}
    S(n-1,e^{n-1} + ye^i) &= Q S(n,e^n+ye^i) Q^{-1}\\
    S(i,e^i + xe^{n-1}) &= Q S(i,e^i + xe^n) Q^{-1}.
\end{align*}
The proof goes on from $n-1$ down to $2$, thus generating any possible transvection.\\
\\
Now suppose $q$ is even. Any instruction in $\SL(2,2^m)$ is an element of order two, and hence any group generated by two instructions is dihedral. However, $\SL(2,2^m)$ is not a dihedral group for $m \ge 2$ and hence cannot be generated by two instructions. We now prove it can be generated by three instructions.

We recall from Dickson's theorem \cite{Dic01} that the maximal
subgroups of $\SL(2,2^m)$ are each isomorphic to either
\begin{itemize}
\item a stabiliser of a one dimensional subspace in the action on the $2^m+1$ subspaces of
$\GF(2^m)^2$ on which (P)SL$(2,2^m)$ naturally acts;
\item a subfield subgroup;
\item a dihedral group of order $2(2^m\pm1)$.
\end{itemize}

Consider the matrices 
$$
	A:=\begin{pmatrix}
	1&0\\
	x&1
	\end{pmatrix}, \quad
	B:=\begin{pmatrix}
	1&x\\
	0&1
	\end{pmatrix}, \quad
	C:=\begin{pmatrix}
	1&x^2\\
	0&1
	\end{pmatrix}
$$
where $x\in\GF(2^m)$ is contained in no proper subfield. Let
$H$ be the subgroup generated by the matrices $A$ and $B$. By the
same arguments as the case $\SL(2,q)$ with $q$ odd we know that $H$
is contained in neither a subspace stabilizer nor a subfield
subgroup and so the only maximal subgroups containing $H$ must be
dihedral of order $2(q\pm1)$. Note that since these are dihedral
groups of twice odd order these subgroups cannot contain pairs of
involutions that commute. Since $BC=CB$ it follows that $C$ cannot
be contained in any of these dihedral subgroups and so no maximal
subgroup contains all of $A$, $B$ and $C$, hence they must generate
the whole group.\\
\\
The base case of the induction thus occurs for $n=3$. Let $x$ such that
$\begin{pmatrix}
1&0\\
x&1
\end{pmatrix}$,
$\begin{pmatrix}
1&x\\
0&1
\end{pmatrix}$ and
$\begin{pmatrix}
1&x^2\\
0&1
\end{pmatrix}$ generate $\SL(2,2^m)$. We shall prove that the matrices
$$
    M_1 := \begin{pmatrix}
    1 & 1 & 0\\
    0 & 1 & 0\\
    0 & 0 & 1
    \end{pmatrix},
    M_2 := \begin{pmatrix}
    1 & 0 & 0\\
    0 & 1 & x\\
    0 & 0 & 1
    \end{pmatrix},
    M_3 := \begin{pmatrix}
    1 & 0 & 0\\
    0 & 1 & 0\\
    x & 0 & 1
    \end{pmatrix}
$$
generate $\SL(3,2^m)$. Denoting
$$
    N_1 := M_1^{-1}M_2^{-1}M_1M_2
    = \begin{pmatrix}
    1 & 0 & x\\
    0 & 1 & 0\\
    0 & 0 & 1
    \end{pmatrix}, \quad
    N_2 := M_2^{-1} M_3^{-1} M_2 M_3
    = \begin{pmatrix}
    1 & 0 & 0\\
    x^2 & 1 & 0\\
    0 & 0 & 1
    \end{pmatrix},
$$
we obtain
$$
    P_3 := N_2^{-1}N_1^{-1}N_2N_1 = \begin{pmatrix}
    1 & 0 & 0\\
    0 & 1 & x^3\\
    0 & 0 & 1
    \end{pmatrix}.
$$
Since
$$    P_3^{-1}M_3^{-1}P_3 M_3 = \begin{pmatrix}
    1 & 0 & 0\\
    x^4 & 1 & 0\\
    0 & 0 & 1
    \end{pmatrix},
$$
we can proceed as above to obtain $S(2,(0,1,x^5))$. We may repeat this process until we derive $S(2, (0,1,x^{2m+1})) = S(2,(0,1,x^2))$, which together with $M_2$ and
$$
    M_3^{-1}M_1^{-1} M_3 M_1
    = \begin{pmatrix}
    1 & 0 & 0\\
    0 & 1 & 0\\
    0 & x & 1
    \end{pmatrix}
$$
generate $\SL(2,2^m)$ acting on the last two coordinates. It is then easy to show that any transvection of the form $S(1,e^1 + ae^i)$ or $S(i,e^i + ae^1)$ for any $i = 2,3$ and any $a \in \GF(2^m)$ can be generated. Thus, the whole special linear group is generated.\\
\\
We now prove the inductive step. More specifically, we show that $\SL(n,q)$ is generated by the following set of instructions:
$$
    \{S(i,v^i) : 1 \le i \le n-2\} \cup \{S(n-1,e^{n-1} + x e^n), S(n,e^n + x e^1)\}.
$$
Again, we can easily generate $S(1,e^1 + x e^n)$ and hence $\SL(3,2^m)$ acting on the coordinates $1$, $n-1$, and $n$. In particular, $S(n-1,e^{n-1} + x e^1)$ is generated and by induction hypothesis we obtain $\SL(n-1,2^m)$ acting on the first $n-1$ coordinates. Finally, any transvection of the form $S(n,e^n + ae^i)$ or $S(i,e^i + ae^n)$ for any $i \le n-1$ and any $a \in \GF(2^m)$ can be easily generated. Thus, the whole special linear group is generated.
\end{proof}


We now turn to the general linear group.

\begin{theorem} \label{th:GL}
The group $\GL(n,q)$ is generated by $n$ instructions for any $n$ and any prime power $q$.
\end{theorem}

\begin{proof}
The proof is split into two parts, depending on the parity of $q$; the even part goes by induction on $n$.
If $q$ is even, we prove that $\GL(n,2^m)$ is generated by the $n$ instructions 
$$
	\{S(i,v^i) :2 \le i \le n-1\} \cup \{S(1,\alpha e^1 + e^2), S(n,\alpha e^1 + e^n)\}
$$ 
for any primitive element $\alpha$. Since $\mbox{det}(S(1,\alpha e^1 + e^n)) = \alpha$, we only need to generate the special linear group. 

For $n=2$, denote $M_i = S(i,(\alpha,1))$ for $i=1,2$. Then we can generate the transposition matrix as follows:
$P = \begin{pmatrix} 0 & 1\\ 1 & 0 \end{pmatrix} = M_1 M_2 M_1^{-1}.$
Since $S(1,(1,\alpha)) = P M_2 P$, we easily generate $S(1,(\alpha,0)) = M_1^{-1} S(1,(1,\alpha)) M_1^2.$
Any transvection $S(1,(1,\alpha^k))$ can then be expressed as
$$
    \begin{pmatrix} 1 &\alpha^k\\ 0 & 1 \end{pmatrix} = 
	\begin{pmatrix} \alpha^{k-1} & 0\\ 0 & 1 \end{pmatrix}
	\begin{pmatrix} 1 &\alpha\\ 0 & 1 \end{pmatrix}
	\begin{pmatrix} \alpha^{-k+1} &0\\ 0 & 1 \end{pmatrix},
$$
and any other transvection is obtained by conjugating by $P$.\\
\\
We now prove the inductive part. We can easily generate $S(1,\alpha e^1 + e^n)$, which combined with $S(n,e^n + \alpha e^1)$ generates $\GL(2,q)$ acting on the coordinates $1$ and $n$. In particular, we obtain the matrix $Q = \begin{pmatrix} 0 & 1\\ \alpha & 0 \end{pmatrix}$, and
$$
    S(n-1,\alpha e^1 + e^{n-1}) = Q^{-1} S(n,\alpha e^1 + e^n) Q.
$$
We then have the complete set of generators for $\GL(n-1,q)$ acting on coordinates $1$ to $n-1$. It is then easy to prove that any transvection of the form $S(i,e^i + ae^n)$ and $S(n,e^n + ae^i)$ for any $1 \le i \le n-1$ and any $a \in \GF(q)$ can be generated.

If $q$ is odd and $n=2$, consider the matrices
$A:=\left(\begin{array}{cc}
1&1\\
0&1
\end{array}\right)$, $B:=\left(\begin{array}{cc}
1&0\\
1&x
\end{array}\right)$ where $x\in\mbox{GF}(q)$ is not contained in any proper subfield.
Arguments analogous to those used in the SL$(2,q)$ case show that $\langle A,A^B\rangle=\mbox{SL}(2,q)$.\\
\\
If $n>2$, we rely on the proof of Theorem \ref{th:SL} for $q$ odd. We know that there exist $x,y \in \GF(q)$ such that $\SL(n,q)$ is generated by
$$
    \{S(i,v^i) : 1 \le i \le n-2\} \cup \{S(n-1,e^{n-1} + xe^n), S(n, e^n + ye^1)\}.
$$
Let $a$ be a primitive element of $\GF(q)$ and $b := (a-1)x/2$. We shall prove that replacing the instruction updating coordinate $n-1$ by $T = S(n-1,ae^{n-1} + be^n)$ in the set above yields a generating set for $\GL(n,q)$. We only need to show that $S(n-1,e^{n-1} + xe^n)$ is generated. We have $T^{(q-1)/2} = S(n-1,-e^{n-1} - xe^n)$ and hence we can easily generate $S(1,e^1 + xe^n)$ and the whole of $\SL(2,q)$ acting on coordinates $1$ and $n$. In particular, we obtain $Q = \mbox{diag}(2^{-1},1,\ldots,1,2)$, whence
$$
    S(n-1,e^{n-1} + x e^n) = S(n-1,-e^{n-1} - xe^n) Q^{-1} S(n-1,-e^{n-1}-xe^n) Q.
$$
Only displaying rows and columns indexed $1,n-1,n$ the equation above reads
$$
	\begin{pmatrix} 1 & 0 & 0\\ 0 & 1 & x\\ 0 & 0 & 1 \end{pmatrix} =
	\begin{pmatrix} 1 & 0 & 0\\ 0 & -1 & -x\\ 0 & 0 & 1 \end{pmatrix}
	\begin{pmatrix} 2 & 0 & 0\\ 0 & 1 & 0\\ 0 & 0 & 2^{-1} \end{pmatrix}
	\begin{pmatrix} 1 & 0 & 0\\ 0 & -1 & -x\\ 0 & 0 & 1 \end{pmatrix}
	\begin{pmatrix} 2^{-1} & 0 & 0\\ 0 & 1 & 0\\ 0 & 0 & 2 \end{pmatrix}.
$$
\end{proof}

We conclude this section by noticing that Theorems \ref{th:SL} and \ref{th:GL} have implication on some classical semigroups of matrices. Denote the semigroup of singular matrices in $\GF(q)^{n \times n}$ as $\mbox{Sing}(n,q)$ and consider the general linear semigroup (also called full linear monoid \cite{Okn98}) and special linear semigroup:
\begin{align*}
	\mbox{GLS}(n,q) &= \GL(n,q) \cup \mbox{Sing}(n,q),\\
	\mbox{SLS}(n,q) &= \SL(n,q) \cup \mbox{Sing}(n,q).
\end{align*}
Note that $\mbox{Sing}(n,q)$ is not an internally computable semigroup. Indeed, the kernel of any singular instruction matrix only contains vectors with Hamming weight equal to zero or one. Thus any matrix whose kernel forms a code with minimum distance at least two cannot be computed by a program only consisting of singular instructions. For instance, the square all-ones matrix of any order over any finite field cannot be computed in that fashion. 

However, according to Theorems 6.3 and 6.4 in \cite{Rus95}, any generating set of $\GL(n,q)$ ($\SL(n,q)$ respectively) appended with any matrix of rank $n-1$ in $\mbox{Sing}(n,q)$ generates $\mbox{GLS}(n,q)$ ($\mbox{SLS}(n,q)$ respectively). Since any singular instruction has rank $n-1$, we conclude that these semigroups are internally computable, and in particular $\mbox{GLS}(n,q)$ is generated by $n+1$ instructions, while $\mbox{SLS}(n,q)$ is generated by $n+1$ instructions unless $q = 2^m$ and $n=2$, where it is generated by four instructions.



\end{document}